\documentclass{llncs}

\usepackage{verbatim}
\usepackage{amsmath}
\usepackage{amssymb}
\usepackage{amsfonts}
\usepackage{latexsym}
\usepackage{amscd}
\usepackage{graphicx}
\usepackage{xcolor}
\usepackage{url}
\usepackage{pb-diagram}
\usepackage{longtable}
\usepackage{soul}
\usepackage{hyperref}
\usepackage{cite}

\newcommand{\Ac}{\mathcal{A}}

\newcommand{\NN}{\mathbb{N}}

\newcommand{\nn}{\mathbb{N}}

\newcommand{\Cc}{\mathcal{C}}
\newcommand{\Pc}{\mathcal{P}}
\newcommand{\Bc}{\mathcal{B}}
\newcommand{\Dc}{\mathcal{D}}
\newcommand{\Qc}{\mathcal{Q}}
\newcommand{\Ic}{\mathcal{I}}
\newcommand{\I}{\mathbf{I}}
\newcommand{\G}{\mathbf{G}}
\newcommand{\B}{\mathbf{B}}
\newcommand{\A}{\mathbf{A}}
\newcommand{\C}{\mathbf{C}}

\newcommand{\y}{{\bf{y}}}
\newcommand{\x}{{\bf{x}}}

\newcommand{\bldy}{{\bf{y}}}
\newcommand{\bldx}{{\bf{x}}}

\newcommand\scalemath[2]{\scalebox{#1}{\mbox{\ensuremath{\displaystyle #2}}}}

\begin{document}

\sloppy

\title{Constructions and Bounds for Batch Codes with Small Parameters}
\author{Eldho K. Thomas \and Vitaly Skachek}
   \institute{Institute of Computer Science\\
     University of Tartu, Estonia\\
     \email{eldho.thomas@ut.ee, vitaly.skachek@ut.ee}}


\maketitle

\begin{abstract}
Linear batch codes and codes for private information retrieval (PIR) with a query size $t$ and a restricted size $r$ of the 
reconstruction sets are studied. New bounds on the parameters of such codes 
are derived for small values of $t$ or $r$ by providing corresponding constructions.
By building on the ideas of Cadambe and Mazumdar, a new bound in a recursive form is derived for batch codes and PIR codes. 

\end{abstract}
%
%
\section{Introduction}
 Batch codes are proposed in \cite{ISHAI} for load balancing in the distributed server systems. They can be broadly classified as linear batch codes and combinatorial batch codes. A particular version of the former is known as switch codes and were mainly studied in \cite{chee, mao,wang} in the context of network switches. Some works on combinatorial batch codes can be found in \cite{BKM, BT,ruj}. 

Locally repairable codes (LRC codes), or codes with locality, which are deeply studied in \cite{CM,FY, GHSY, RMV, RPDV}, share lots of similarities with batch codes, and therefore many of the properties of these two code families are expected to be related to each other. In \cite{ZS},  new upper bounds on the parameters of batch codes based on the classical Singleton bound which do not depend on the size of the underlying alphabet are derived. 
Batch codes turn out to be a special case of private information retrieval (PIR) codes~\cite{FVY}. Indeed, PIR codes support only queries of type $(x_i, x_i,\ldots , x_i), 1 \le i \le k$, whereas batch codes support queries of a more general form 
$(x_{i_1} , x_{i_2}, \ldots, x_{i_t} )$, possibly for different indices $i_1, i_2, \ldots , i_t$. It follows that batch codes can be used as PIR codes. In \cite{VY}, batch codes with unrestricted size of reconstruction sets are considered and some bounds on the optimal length of batch and PIR codes for a given batch size and dimension are proposed.

In this work, we construct new families of batch codes with restricted size $r$ of reconstruction sets. We also generalize the existing bounds on the dimension of LRC codes proposed in \cite{RMV} to batch codes using the connections between the two families. This paper is organized as follows.  In Section \ref{sec:three}, we propose an optimal construction of batch codes with $t=2, r\ge 2,$ and a construction with $t \ge 3, r=2$. In Section~\ref{sec:four}, we present constructions of PIR codes with arbitrary $t$ and $r$ for $r|k$. In Section~\ref{sec:ab}, we derive a new upper bound on the dimension $k$ of batch codes. 

%
%
\section{Notations and related works}
We start with introducing some notations. Denote by $\NN$ the set of nonnegative integers. For $n \in \NN$ we denote $[n]=\{1, 2,\ldots, n\}$. A $k\times k$ identity matrix will be denoted by $\I_k$, all-one column vector by ${\bf 1}$. We use ${\bf 0}$ to denote an all-zero column vector and a zero matrix. The right dimensions will be clear from the context. 
Let $\bldx$ be a vector of length $n$ indiced by $[n]$. Take $S \subseteq [n]$. Then $\bldx_S$ stands for a sub-vector of $\bldx$ indiced by $S$. If $\A$ is a matrix, then $\A^{[i]}$ denotes the $i$-th column in $\A$. 

Let $\Qc$ be a finite alphabet. Consider an information vector $\x = (x_1, x_2, \ldots , x_k) \in \Qc^k$. 
The code is a set of vectors $\{\y = \Cc(\x) \; | \;  \x \in 	\Qc^k \} \subseteq \Qc^n$,
where $\Cc: \Qc^k \rightarrow \Qc^n$ is a bijective mapping, and $n \in \NN$. By slightly abusing the notation, $\Cc$
will also be used to denote the above code. 

In this work, we study (primitive, multiset) batch codes with restricted size of the recovery sets, 
as they are defined in~\cite{ZS} (see also~\cite{VY}).  

\begin{definition}
An $(n,k, r, t)$ batch code $\Cc$ over a finite alphabet $\Qc$ is defined by
an encoding mapping $\Cc \; : \; \Qc^k \rightarrow \Qc^n$, and a decoding mapping $\Dc \; : \; \Qc^n \times [k]^t\rightarrow \Qc^t$, such that  
\begin{enumerate}
\item
For any $\bldx \in \Qc^k$ and a multiset $(i_1, i_2, \cdots, i_t) \subseteq [k]^t$, 
\[
\Dc\left(\bldy = \Cc(\bldx), i_1, i_2, \cdots, i_t\right) = (x_{i_1}, x_{i_2}, \cdots, x_{i_t}) \; .
\]
\item 
The symbols in the query $(x_{i_1}, x_{i_2}, \cdots, x_{i_t})$ can be reconstructed from 
$t$ respective disjoint recovery sets of symbols of $\bldy$ of size at most $r$ each
(the symbol $x_{i_\ell}$ is reconstructed from the $\ell$-th recovery set for each $\ell$, $1 \le \ell \le t$). 
\end{enumerate}
\label{def:batch}
\end{definition}
If the alphabet $\Qc$ is a finite field, and the associated encoding mapping $\Cc \; : \; \Qc^k \rightarrow \Qc^n$
is linear over $\Qc$, the corresponding code is termed linear. In that case, for each fixed $(i_1, i_2, \cdots, i_t) \in [k]^t$, the corresponding decoding mapping $\Dc$ from $\Qc^n$ to $\Qc^t$ is linear over $\Qc$ too. Additionally, 
if the encoding mapping $\Cc \, : \, \bldx \mapsto \bldy$ is such that $\bldx$ is a sub-vector of $\bldy$,
then the corresponding code is called {\it systematic}.

This setup has first appeared in~\cite{LS}. It was shown therein that the minimum distance $d$ of a batch code satisfies $d \ge t$. 
It is worth mentioning that batch codes are closely related to locally repairable codes, 
which have been extensively studied in the context of the distributed data storage.  
The main difference between them is that in batch codes we are interested in the reconstruction of information symbols in 
$\bldx$, while in locally repairable codes we are interested in the recovery of coded symbols in~$\bldy$. 

\medskip
The following property of batch codes is stated as Corollary III.2 in \cite{ZS}.  
\begin{lemma}
Let $\Cc$ be a linear $(n,k,r,t)$ batch code over $\Qc$, and $\bldx \in \Qc^k$, whose encoding is $\bldy \in \Cc$. Let $R_1, R_2, \cdots, R_t \subseteq [n]$ be $t$ disjoint recovery sets for the coordinate $x_i$. Then, there exist indices $a_2 \in R_2$, $a_3 \in R_3$, $\cdots$, $a_t \in R_t$, such that 
if we fix the values of all coordinates of $\bldy$ indexed by the sets $R_1, R_2 \backslash \{ a_2 \}, R_3 \backslash \{ a_3 \}, \cdots,  R_t \backslash \{ a_t \}$, then the values of the coordinates of $\bldy$ indexed by $\{a_2, a_3, \cdots, a_t \}$ are uniquely determined. 
\label{cor:sets}
\end{lemma}

There is a number of bounds on the parameters of batch codes in the literature, but it is often
difficult to make a comparison due to slight variations in the models and assumptions made.
Thus, it is proven in \cite{ZS} that for a linear $(n,k,r,t)$ batch code over $\Qc$,
\begin{equation}
\label{nonsyst}
n \ge k + d + \max_{1\le \beta \le t}\left \{(\beta - 1)\left( \left\lceil \frac{k}{r\beta - \beta + 1} \right \rceil -1 \right )\right \}-1 \; .
\end{equation}
In particular, when the code is systematic, the bound can be tighten a bit, as follows:
\begin{equation}
\label{eq:bound}
n \ge k + d + \max_{2\le \beta \le t}\left \{(\beta - 1)\left( \left\lceil \frac{k}{r\beta - \beta-r + 2} \right \rceil -1 \right )\right \}-1 \; .
\end{equation}
If the queries in Definition~\ref{def:batch} are restricted to $i_1 = i_2 = \cdots = i_t$, then the corresponding code is called an
$(n,k, r, t)$ code for private information retrieval (PIR)~\cite{FVY}, or simply $(n,k, r, t)$ PIR code. In particular, all batch codes are PIR codes with the corresponding parameters. It should be mentioned that 
the proofs of (\ref{nonsyst}) and (\ref{eq:bound}) in \cite{ZS} work in analogous way for the PIR codes too, and therefore 
these two bounds hold for general and systematic $(n,k,r,t)$ PIR codes, respectively.

Systematic linear $(n,k,r,t)$ PIR codes can be viewed as LRC codes with \emph{locality of information symbols} and availability~\cite{RPDV}. A number of bounds on the parameters of the latest family were derived in~\cite{RPDV}, and in subsequent works. 
Specifically, when re-written for systematic batch code setting, the following bound holds:  
\begin{equation}
\label{eq:dima1}
d + k+ \left \lceil  \frac{(t-1)(k-1)+1}{(t-1)(r-1)+1} \right\rceil -2 \le n.
\end{equation}
A comparison between~(\ref{eq:dima1}) and~(\ref{eq:bound}) is not always straightforward, in particular due to the minimization term in~(\ref{eq:bound}).  However, with the aid of a computer, we verified that the bound~(\ref{eq:dima1}) gives equal or slightly higher values of $n$ compared to~(\ref{eq:bound}) for small values of $d,k,r$ and $t$. Hereafter, 
we employ the bound~(\ref{eq:dima1}) in the analysis of the optimal values of $n$. However, this bound is not always tight, especially for small alphabets. 

Binary simplex codes of length $n = 2^m-1$ are shown to be optimal batch codes with parameters $k=m$, $t = 2^{m-1}- 2$ and $r=2$  (for any $m \in \nn$)~\cite{mao}, yet those codes exist only for very specific parameters.

The LRC codes were extensively studied in the last years. Thus, in \cite{tamobarg}, a lower bound on the length $n$ is presented for an LRC with locality \emph{of all symbols} and availability. It should be noted that codes with locality (and availability) of all symbols are a special case of codes 
with locality (and availability) of information symbols, and therefore the bounds derived for the former family are not directly applicable to the latter family. 

It is shown in~\cite{tamobarg} that it is possible to construct a $t$-fold power of the binary $(r+1,r)$ single parity check code in order to obtain an LRC with availability (termed \emph{direct product code}) for specific parameters. 
An algebraic construction of binary LRC codes in~\cite{WZ} further improves on the rate of the direct product code. However, in general, the resulting construction is non-systematic, and thus it is not straightforward how to derive an analogous result for batch/PIR codes, see~\cite[Example 5]{Skachek-survey}. It would be interesting to extend those techniques to batch/PIR codes, 
but that is left out of the scope of this paper.

In Section \ref{sec:four}, we present constructions of binary PIR codes for arbitrary $r$ and $t$ achieving rate $\frac{r}{r+t-1}$ for $k \ge r^2$ similar to their counterparts in~\cite{WZ}. For $t \in \{2,3,4\}$ these codes are batch codes.
 The achieved rate is close to optimal, especially for small values of $t$ and $r$.

A special case of $(n,k, r, t)$ batch and PIR codes, where the size of the recovery sets $r$ is not restricted (for example, it can be assumed that $r =n$), is studied in \cite{VY}. Let $\Bc(k, t)$ be the shortest length $n$ of any systematic linear batch code with unrestricted size of the recovery set, and $\Pc(k, t)$ be the shortest length $n$ of any linear systematic PIR code with unrestricted size of the recovery set. Then the optimal redundancy of batch and PIR codes, respectively, is defined as
$\gamma_{\Bc}(k, t)= \Bc(k, t) - k$ and  $\gamma_\Pc(k, t)= \Pc(k, t) - k$.

\begin{proposition}\cite{VY} It holds $\Bc(k, t) = \Pc(k, t)$ for $1 \le t \le 4$, and
$\gamma_{\Bc}(k, t) \le \gamma_{\Pc}(k, t) + 2\lceil \log(k) \rceil \cdot \gamma_{\Pc}(k/2, t - 2)$
 for $5  \le t \le 7$.
\end{proposition}
Hereafter, we denote the optimal length of a linear {\it systematic} batch code and PIR code with the size of the reconstruction sets $r$ as $\Bc(k,r,t)$ and $\Pc(k,r,t)$, respectively.
\medskip

%
%
\section{Batch codes with $r = 2$ or $t=2$}
\label{sec:three}
\subsection{Optimal batch codes with $r \ge 2$ and $t=2$}
In this and subsequent sections, we construct $(n,k,r,t )$ batch codes for specific values of $t$ and $r$.
To this end, consider an $(n,k,r,t)$ systematic batch code with $r\ge 2$ and $t=2$.  
Then, by using $d \ge t$, from the bound in (\ref{eq:dima1}), we have
\begin{equation}
\label{eq:bound2} 
n \ge  \left \lceil \frac{k}{r}\right\rceil+ k \; .
\end{equation}
The construction of codes attaining this bound, for $t=2$ and $r=2$, is presented in~\cite[Example 2]{ZS}. In the sequel,
we generalize that construction to other values of $t$ and $r$. 
\medskip

First, we show that the bound~(\ref{eq:bound2}) is optimal for
$t=2$ and any $r \ge 2$. We achieve that by constructing corresponding $(n,k,r,t)$ batch codes.

Take $\mathbf{G}$ to be a $k\times n$ binary systematic generator matrix of a code $\Cc$ defined as follows:
\begin{equation}
\mathbf{G}=
\left (\begin{array}{ccccc|ccccc}
\mathbf{I}_r & {\bf 0} & \cdots  & {\bf 0} & {\bf 0} & {\bf 1} & {\bf 0} & \cdots  & {\bf 0} & {\bf 0}  \\
{\bf 0} & \mathbf{I}_r & \cdots  & {\bf 0} & {\bf 0} & {\bf 0} & {\bf 1} & \cdots  & {\bf 0} & {\bf 0}  \\
\vdots & \vdots & \ddots & \vdots  &\vdots & \vdots &  \vdots &\ddots &  \vdots &\vdots  \\
{\bf 0} &  {\bf 0} & \cdots  & \mathbf{I}_r & {\bf 0} &{\bf 0} & {\bf 0} & \cdots  & {\bf 1} & {\bf 0} \\
{\bf 0} &  {\bf 0} & \cdots  & {\bf 0} & \mathbf{I}_s & {\bf 0} & {\bf 0} & \cdots  & {\bf 0} & {\bf 1}  \\
\end{array} \right) \; ,
\label{eq:optimal_2_r} 
\end{equation}
where $s= k \mod r$, and recall that ${\bf 1}$ denotes all-one column vector. 

It is easy to see that $\Cc$ supports any query of size $t = 2$. \\
If $r|k$, then $n = \frac{k}{r}(r+1)$, which satisfies the lower bound (\ref{eq:bound2}) with equality.
When $ r\nmid k$, we have $s= k-\left \lfloor \frac{k}{r} \right \rfloor r$, and 
\begin{eqnarray*}
n &=& \left \lfloor \frac{k}{r} \right \rfloor (r+1)+ s+1 = \left \lceil \frac{k}{r} \right \rceil +k \; ,
\end{eqnarray*}
which also satisfies~(\ref{eq:bound2}) with equality.
\medskip

Since $\Bc(k,r,t) \ge \Pc(k,r,t)$, we can summarize the result as in the following proposition.
\begin{proposition}
\label{pro:t2r}
For any $k$, $t=2$ and $r\ge 2$,  
$$\Bc(k,r,t) = \Pc(k,r,t)=\left \lceil \frac{k}{r} \right \rceil + k \; .$$
\end{proposition} 
\begin{corollary}
For any $k \ge 1$, $t=2$ and $r \ge 2$, the optimal length of a non-systematic batch and PIR code, $\Bc_n(k,r,t)$ and $\Pc_n(k,r,t)$, respectively, satisfies
$$\left \lceil \frac{k}{r} \right \rceil + k \ge \Bc_n(k,r,t) \ge \Pc_n(k,r,t) \ge \left \lceil \frac{k}{2r-1} \right \rceil +k \; .$$
\end{corollary}
The right-most inequality is obtained from the bound (\ref{nonsyst}) by substituting $d\ge t=2$ and $\beta = 2$. 

\subsection{Batch codes and PIR codes with $t \ge 2$ and $r=2$}
In this section, we propose a construction of (systematic) binary PIR codes such that
\begin{equation}
\label{eq:restriction-t}
2\le t \le \max\left\{\left \lceil\frac{k}{r}\right\rceil, r\right\}+2 
\end{equation}
and $r=2$. We achieve this by using a generator matrix with columns of weight at most $2$. The constructed codes are batch codes for $t=2,3,4$.  In particular, for $t = 2$ and $r=2$, the construction is identical to the one in the previous section. 

Let $\mathbf{G}$ be a binary generator matrix defined as $\mathbf{G} = [ \, \I_k \, | \, \A \,]$. 
In the sequel we describe how to construct the sub-matrix $\A$.

When $k$ is even or $t-1$ is even, the sub-matrix $\A$ has all its columns of weight 2 and all its rows of weight $t-1$, such that there is no $1$-square pattern. In other words, for any $i_1, i_2, j_1, j_2$, $i_1 \neq i_2$, $j_1 \neq j_2$, at least one of the entries $A_{i_1,j_1}$, 
$A_{i_1,j_2}$, $A_{i_2,j_1}$ and $A_{i_2,j_2}$ in $\A$ is zero.

The total number of columns in $\mathbf{G}$ is $n= k+ (t-1) \frac{k}{2}$.
In particular, for $r=2$ and $t=2$, this bound coincides with~(\ref{eq:bound2}).
We remark that the above construction requires some modification when $k=2$ and $3$. This is due to the fact that it is 
impossible to avoid 1-squares if the number of rows in $\mathbf{G}$ is two or three (for $k=2$, $t=3,4$, the corresponding values of $n$ are $5$ and $7$, respectively and for $k=3$, $t=3,4$, the values of $n$ are $6$ and $9$ respectively). 
Moreover, the right-hand inequality in~(\ref{eq:restriction-t}) is required in order to fit $t-2$ ones per row, while avoiding 1-squares.

When $k$ is odd and $t-1$ is odd, then $\A$ has all (except the last) columns of weight 2, and the last column of weight 1. All its rows are of weight $t-1$, and there is no $1$-square pattern in $\A$. 
Therefore the total number of columns in $\G$ is 
\[
n = k + \left \lceil (t-1) \cdot \frac{k}{2} \right\rceil \; .  
\]



\begin{proposition}
The code $\Cc$ defined by the above generator matrix $\mathbf{G}$ for $t=3$ supports any query of the form $(x_i, x_j, x_\ell)$ with recovery sets of size at most $2$, $i, j, \ell \in [k]$.  
\end{proposition}

\begin{proposition}
The code $\Cc$ defined by the above generator matrix $\mathbf{G}$ for $t=4$ supports any query of the form $(x_i, x_j, x_\ell, x_h)$ with recovery sets of size at most $2$, $i, j, \ell, h \in [k]$.  
\end{proposition}

\begin{proposition}
The code $\Cc$ defined by the above generator matrix $\mathbf{G}$ for general $t \ge 5$ supports any query of size $t$ of the form $(x_i, x_i, \cdots, x_i)$ with recovery sets of size at most $2$.  
\end{proposition}

\begin{proposition}
\label{pro:tr2}
For $r=2$ and $3\le t \le \max\{\left \lceil\frac{k}{r}\right\rceil, r\}+2$, 
\begin{equation*} 
k+t-2+\left \lceil\frac{t -1)(k-1)+1}{t}\right\rceil 
\le \Pc(k,r,t) \le k + \left \lceil (t-1) \cdot \frac{k}{2} \right\rceil \; . 
\end{equation*}
\end{proposition}
The left-most inequality in Proposition \ref{pro:tr2} is obtained from (\ref{eq:dima1}) by substituting $r=2$. 
\begin{proposition}
\label{pro:tr_3_4}
For $r=2$ and $t \in  \{ 3, 4 \}$, 
\begin{equation*} 
k+t-2+\left \lceil\frac{t -1)(k-1)+1}{t}\right\rceil 
\le \Bc(k,r,t) \le k + \left \lceil (t-1) \cdot \frac{k}{2} \right\rceil \; . 
\end{equation*}
\end{proposition}
The following examples illustrate the above constructions.
\begin{example}
Consider a binary $(n,k,r,t)$ batch (PIR) code $\Cc$ with $k=5$, $t=3 $ and $r=2$. From Proposition~\ref{pro:tr_3_4}, we have 
$9 \le \Bc(5,2,3) \le 10$. We construct a batch (PIR) code $\Cc$ of length $10$ with the above parameters using the following $5 \times 10$ generator matrix:
$$
\mathbf{G}=
\left (\begin{array}{ccccc}
1 &0 &0 &0 &0     \\
0 &1 &0 &0 &0     \\
0 &0 &1 &0 &0  \\
0 &0 &0 &1 &0     \\
0 &0 &0 &0 &1 
\end{array} \right |
\left. \begin{array}{ccccc}
1  &0  &  0 & 1 & 0   \\
0  &1  &  0 & 0 & 1   \\
0  &0  &  1 & 1 & 0   \\
0  &0  &  1 & 0 & 1   \\
1  &1  &  0 & 0 & 0 
\end{array} \right) \; .
$$
\end{example}
\begin{example}
Take a binary $(n,k,r,t)$ batch (PIR) code $\Cc$ with $k=5$, $t=4 $ and $r=2$. From Proposition~\ref{pro:tr_3_4}, we have 
$11 \le \Bc(5,2,4) \le 13$. We construct a batch (PIR) code $\Cc$ of length $13$ with the above parameters using the following $5 \times 13$ generator matrix:
$$
\mathbf{G}=
\left (\begin{array}{ccccc}
1 &0 &0 &0 &0     \\
0 &1 &0 &0 &0     \\
0 &0 &1 &0 &0  \\
0 &0 &0 &1 &0     \\
0 &0 &0 &0 &1 
\end{array} \right |
\left. 
\begin{array}{cccccccc}
1  &0  &  0 & 1 & 0 & 1 & 0 & 0  \\
0  &1  &  0 & 0 & 1 & 0 & 1 & 0 \\
0  &0  &  1 & 1 & 0 & 0 & 1 & 0 \\
0  &0  &  1 & 0 & 1 & 1 & 0 & 0 \\
1  &1  &  0 & 0 & 0 & 0 & 0 & 1
\end{array} \right ) \; .
$$
\end{example}
%
%

\section{PIR codes for arbitrary $t>2$ and $r>2$}
\subsection{Case $r|k$}
\label{sec:four}
In this section, by using a generator matrix with columns of weight at most $r$, we generalize the construction in Section~\ref{sec:three} to systematic PIR codes with arbitrary 
parameters $2< t\le\max\{\frac{k}{r},r\}+2$ and $r>2$, $r \mid k$. The corresponding upper bounds are implied by the construction.
 
Let $\G$ be a generator matrix of the form $ \mathbf{G} = [\I_k \, | \, \mathbf{A} \, | \, \mathbf{B}]$,
where $\I_k$ is the systematic part. Here, 
$\mathbf{A}$ is a $k \times \frac{k}{r}$ matrix with the $j$th column of the form 
$\mathbf{A}^{[j]} = (a_{1,j}, a_{2,j}, \cdots, a_{k,j})^T$, $1 \le j \le \frac{k}{r}$, where $a_{i,j} = 1$ if $(j-1)r+1 \le i \le jr$, and $a_{i,j} = 0$ otherwise. The matrix $\mathbf{B}$ is defined as follows. Each column has weight $\min\{r,k/r\}$, 
every row in $\B$ has weight $t-2$, and there are no $1$-squares in $[\, \mathbf{A} \, | \, \mathbf{B} \,]$. That is, $\mathbf{B}$ is constructed in such a way that all rows of $\G$ have weight $t$, columns have weight at most $r$, and there are no $1$-squares in $[\, \mathbf{A} \, | \, \mathbf{B} \,]$. The absence of $1$-squares is instrumental in finding disjoint recovery sets for all information symbols.

\begin{itemize}
\item{\bf Case 1: $\frac{k}{r} < r$.} 
In that case we choose $\mathbf{B}$ to be a $k \times (t-2)r$ matrix defined as per the rules above. Then,  
$$n = k + \frac{k}{r} + (t-2)r = (r+1)\frac{k}{r} +(t-2)r \; ,$$
and the code rate is $$\frac{k}{n}= \frac{k}{(r+1)\frac{k}{r} +(t-2)r} < \frac{k}{(r+1)\frac{k}{r} +(t-2)\frac{k}{r}} = \frac{r}{r+t-1} \, ,$$
where the inequality is due to $\frac{k}{r} < r$. 

\item{\bf Case 2: $\frac{k}{r} \ge r$.} 
In this case, we can choose $\mathbf{B}$ as a $k \times (t-2)\frac{k}{r}$ matrix as defined above.
To this end,  $$n = (r+1)\frac{k}{r}+ (t-2)\frac{k}{r} \, , $$ and the rate of the code is 
$$\frac{k}{n} = \frac{k}{(r+1)\frac{k}{r} +(t-2)\frac{k}{r}} = \frac{r}{r+t-1} \; .$$

By a suitable choice of $k$, it is always possible to construct a PIR code (or batch code if $t \in \{2,3,4\}$) achieving the rate $\frac{r}{r+t-1}$, which is close to the optimal rate given in \cite{tamobarg}. Note that the condition $t\le \max\{\frac{k}{r},r\}+2$ is necessary to make sure that no 1-squares are generated.
\end{itemize}

\begin{proposition}
For $r>2$ and $2<t\le \zeta+2$ with $r | k$, 
$$\Pc(k,r,t) \le (r+1)\frac{k}{r}+ (t-2)\zeta$$ where $\zeta = \max\{ \frac{k}{r}, r\}$.
\label{prop:t_3}
\end{proposition}

\begin{proposition}
\label{prop:batch_t_3}
Let $t=3$. The code $\Cc$ defined by the generator matrix $\mathbf{G}$ supports any query of the form $(x_i, x_j, x_\ell)$ with recovery sets of size at most $r$, $i, j, \ell \in [k]$. Therefore it is a batch code.   
\end{proposition}

\begin{proposition}
Let $t=4$. The code $\Cc$ defined by the generator matrix $\mathbf{G}$ supports any query of the form $(x_i, x_j, x_\ell, x_h)$ with recovery sets of size at most $r$, $i, j, \ell, h \in [k]$. Therefore it is a batch code.  
\end{proposition}

\begin{proposition}
Let $t \ge 5$. The code $\Cc$ defined by the above generator matrix $\mathbf{G}$ supports any query of size $t$ of the form $(x_i, x_i, \cdots, x_i)$ with recovery sets of size at most $r$. Therefore it is a PIR code.  
\end{proposition}

\begin{example}
Let $k = 8,~ r=4,~ t=3$ so that $k/r= 2$. Then the following generator matrix $\G$ generates a batch codes of length $n=14$.
$$ 
\G = \left (\begin{array}{cc}
\I_4 & {\bf 0}\\
\hline
{\bf 0} & \I_4
\end{array} \right |
\begin{array}{cc}
{\bf 1}  &{\bf 0} \\
\hline
{\bf 0}  &{\bf 1} 
\end{array} 
\left| \begin{array}{c}
\I_4 \\
\hline 
\I_4
\end{array} \right ) \; .
$$ 
\end{example}
\begin {example}
Let $k = 12,~ r=3,~ t=5$ so that $k/r= 4$. Then the following generator matrix $\G$ generates a PIR code of length $n= 12+ 4+ 3\cdot 4=28$.
$$ 
\G = \left (
\begin{array}{cccc}
 & &  &\\
\I_3 & {\bf 0} &  {\bf 0}& {\bf 0}\\
 & &  & \\
\hline
 & &  & \\
 {\bf 0} & \I_3 &  {\bf 0} & {\bf 0}\\
 & &  & \\
\hline
 & &  & \\
{\bf 0}  &  {\bf 0}& \I_3 & {\bf 0}\\
 & &  & \\
\hline
 & &  & \\
{\bf 0}  &  {\bf 0} & {\bf 0}& \I_3\\

 & &  & 
\end{array} \right |
\begin{array}{cccc}
 & &  & \\
{\bf 1}  &{\bf 0}  & {\bf 0} & {\bf 0}\\
 & &  & \\
\hline
& & &\\
{\bf 0}  &{\bf 1}  & {\bf 0} & {\bf 0}\\
& & & \\

\hline

& & &\\
{\bf 0}  &{\bf 0}  &{\bf  1} & {\bf 0}\\
 & &  & \\
\hline
 & &  & \\
{\bf 0}  &{\bf 0}  &{\bf  0} & {\bf 1}\\
 & &  & 
\end{array} 
\left| \begin{array}{cccc}
1 & 1 &1 &0 \\
0 & 0 &0 &0 \\
0 & 0 &0 &0 \\
\hline 
0 & 1 &0 &0 \\
0 & 0 &1 &0 \\
0 & 0 &0 &1 \\
\hline 
1 & 0 &0 &0 \\
0 & 0 &0 &1 \\
0 & 1 &0 &0 \\
\hline 
0 & 0 &0 &1 \\
0 & 0 &1 &0 \\
1 & 0 &0 &0 \\
\end{array} \right |
\begin{array}{cccc}
0 & 0 &0 &0 \\
1 & 1 &1 &0 \\
0 & 0 &0 &0 \\
\hline 
0 & 1 &0 &0 \\
0 & 0 &1 &0 \\
0 & 0 &0 &1 \\
\hline 
0 & 1 &0 &0 \\
1 & 0 &0 &0 \\
0 & 0 &0 &1 \\
\hline 
0 & 0 &1 &0 \\
0 & 0 &0 &1 \\
1 & 0 &0 &0 \\
\end{array} \left |
\begin{array}{cccc}
0 & 0 &0 &0 \\
0 & 0 &0 &0 \\
1 & 1 &1 &0 \\
\hline 
0 & 1 &0 &0 \\
0 & 0 &1 &0 \\
0 & 0 &0 &1 \\
\hline 
0 & 0 &0 &1 \\
0 & 1 &0 &0 \\
1 & 0 &0 &0 \\
\hline 
1 & 0 &0 &0 \\
0 & 0 &0 &1 \\
0 & 0 &1 &0 \\ \end{array} \right ).
$$
\end{example} 
The rate of the above code is 
\[
\frac{k}{n} =  \frac{12}{28} = \frac{3}{7} = \frac{r}{r+t-1} \; .
\] 
This rate is greater than the rate of the direct product code in \cite{tamobarg}.

\subsection{Case $t=3$ and $r\nmid k$}
In the sequel, we extend the results in the previous subsection towards the case where $t=3$ and $r \nmid k$. 
Let $\mathbf{G}$ be the generator matrix of the block form $ \mathbf{G} = [\I_k \, | \, \mathbf{A} \, | \, \mathbf{B} \, | \, \mathbf{C}]$. 

$\mathbf{A}$ is a $k \times \left \lfloor\frac{k}{r}\right \rfloor$ matrix with the $j$th column of the form 
$\mathbf{A}^{[j]} = (a_{1,j}, a_{2,j}, \cdots, a_{k,j})^T$, $1 \le j \le\left \lfloor\frac{k}{r}\right \rfloor$, where $a_{i,j} = 1$ if $(j-1)r  + 1 \le i \le jr$, and $a_{i,j} = 0$ otherwise. 

Denote $s= k \mod r$. The matrix $\B$ is a $k \times (s+1)$ block matrix, defined as 
$\B = \left[ \B_1^T | \B_2^T \right]^T$, 
where $\B_1$ is $(k-s) \times (s+1)$ and $\B_2$ is an $s \times (s+1)$ matrix  $[{\bf 1} \, | \, \I_s]$. 

We take $\tau \triangleq \min\left\{ r-s, \left \lfloor \frac{k}{r} \right \rfloor \right\}$. 
The first column of $\B_1$, $\B_1^{[1]}$, has $\tau$ entries 1, each entry appears in a different block of rows $[(j-1)r +1, jr]$
for $j = 1, 2, \cdots, \left \lfloor\frac{k}{r}\right \rfloor$. 
We take also $\eta \triangleq \min\left\{ r-1, \left \lfloor \frac{k}{r} \right \rfloor \right\}$.
The columns $\B_1^{[2]}, \B_1^{[3]}, \cdots, \B_1^{[s+1]}$,
all have $\eta$ entries 1, each entry appears in a different block of rows. Additionally, every row in $\B_1$
contains at most one non-zero entry. 

We observe that every column in $\B$ has at most $r$ ones.  
Denote $\gamma \triangleq \min \left \{r, \left\lfloor \frac{k}{r}\right \rfloor \right\}.$
The matrix $\C$ is constructed according to the following rules:
\begin{itemize} 
\item
Each column in $\C$ has $\gamma$ ones (except possibly for the last column);
\item
The last $s$ rows in $\C$ are zeros;
\item
Each row in $[\, \mathbf{A} \, | \, \mathbf{B} \, | \, \mathbf{C} \,]$ has two ones; 
\item
There are no $1$-squares in $[\, \mathbf{B} \, | \, \mathbf{C} \,]$.
\end{itemize}

Next, we estimate the total number of columns in $\mathbf{G}$. The number of ones in the first $k-s$ positions of $\B^{[1]}$ is 
$\tau$. The number of ones in the first $k-s$ positions of each of $\B^{[2]}, \B^{[3]}, \cdots, \B^{[s+1]}$ is 
$\eta$. Since there are two ones in each of the first $k-s$ rows of $[\, \A \, | \, \B \, | \, \C \, ]$, the total number of columns in $\G$ is 
$$
n = (r+1) \left \lfloor \frac{k}{r} \right \rfloor +2s+1 +\left \lceil  \frac{ (k-s) -\tau -\eta \cdot s }{\gamma}\right \rceil \; . $$
\begin{proposition} For $t=3$ and $r\ge 3$,
\begin{equation*}
k+1+\left \lceil\frac{2k-1}{2r-1} \right \rceil
 \le \Bc(k,r,t) \le \\ 
\begin{cases}
(r+1)\frac{k}{r} + \zeta & \mbox{ if } r|k \\
(r+1) \left \lfloor \frac{k}{r} \right \rfloor +2s+1 +\left \lceil  \frac{ (k-s) -\tau -\eta \cdot s }{\gamma}\right \rceil & 
\mbox{ if } r \nmid k
\end{cases} 
\end{equation*}
where $\zeta = \max\{ \frac{k}{r}, r\}$.
\end{proposition}

\begin{proposition}
The code $\Cc$ defined in this section by the generator matrix $\mathbf{G}$ supports any query of the form $(x_i, x_j, x_\ell)$ with recovery sets of size at most $r$, $i, j, \ell \in [k]$.  
\end{proposition}

\begin{example}
Let $k = 11,~ r=3,~ t=3$ so that $\lfloor k/r\rfloor= 3$ and $s=2$. Then the following generator matrix $\G$ generates a batch code of length $n= 19$.

$$ 
\G = \left (\begin{array}{cccc}
 & & & \\
\I_3 & {\bf 0} &  {\bf 0} &{\bf 0}\\
 & &  & \\
\hline
 & &  & \\
 {\bf 0} & \I_3 &  {\bf 0}&{\bf 0}\\
 & &   &\\
\hline
 & & &  \\
{\bf 0}  &  {\bf 0}& \I_3 &{\bf 0}\\
 & &  & \\
\hline
& & &\\
{\bf 0}  &  {\bf 0} & {\bf 0}& \I_2
 \end{array} \right |
\begin{array}{ccc}
 & &   \\
{\bf 1}  &{\bf 0}  & {\bf 0} \\
 & &   \\
\hline
& & \\
{\bf 0}  &{\bf 1}  & {\bf 0} \\
& &  \\
\hline
& &\\
{\bf 0}  &{\bf 0}  &{\bf  1}\\
 & &   \\
\hline
 & &   \\
{\bf 0}  &{\bf 0}  &{\bf  0} \\
\end{array}
\left |
\begin{array}{c}
 1   \\
 0   \\
 0   \\
\hline 
 0  \\
 0  \\
 0  \\
\hline 
 0  \\
 0  \\
 0  \\
\hline 
  \\
 {\bf 1}
 \end{array} 
\right|
\begin{array}{cc}
 0 &0  \\
 0 &1  \\
 0 &0  \\
\hline 
 1 &0  \\
 0 &1  \\
 0 &0  \\
\hline 
 1 &0  \\
 0 &0  \\
 0 &0  \\
\hline 
  & \\
  & \I_2
 \end{array} 
\left| \begin{array}{cc}
0 & 0 \\
0 & 0  \\
1 & 0  \\
\hline 
0 & 0  \\
0 & 0  \\
1 & 0  \\
\hline 
0 & 0  \\
0 & 1  \\
1 & 0  \\
\hline 
0 & 0  \\
0 & 0  
\end{array} \right ) \; .
$$
\end{example}
We summarize the bounds on $\Bc(k,r,t)$ and $\Pc(k,r,t)$ for $t \in \{2,3,4\}$ and for $r \ge2 $ in Tables~\ref{table:results_lower} and~\ref{table:results_upper}.

\begin{table*}[htb]
\begin{center}
\scalemath{1.0}{{\begin{tabular}{c| ccc} 
\hline 
& $t = 2$ & $t = 3$ & $t = 4$ \\
\hline
$r \ge 2$ \; & \; $k + \left \lceil k/r \right\rceil$ \; & \; $k+1+\left \lceil\frac{2k-1}{2r-1} \right \rceil$ \; & \; $k+2+\left \lceil\frac{3k-2}{3r-2} \right \rceil$ \\
\hline
\end{tabular}}}
\end{center}
\caption{Lower bounds for $\Bc(k,r,t)$ and $\Pc(k,r,t)$}
\label{table:results_lower}
\vspace{-.25cm}
\end{table*}
\begin{table*}[htb]
\begin{center}
\scalemath{0.92}{\begin{tabular}{c| ccc} 
\hline 
& $t = 2$ & $t = 3$ & $t = 4$ \\
\hline
$r=2$, $k>3$ & $k + \left \lceil k/2 \right\rceil$ & $2k$ & $k + \left \lceil 3k/2 \right\rceil$ \\
$r\ge 3$ & $k + \left \lceil \frac{k}{r} \right\rceil$ & $\begin{cases}
(r+1)\frac{k}{r} + \zeta  & \mbox{ if } r|k \\
(r+1) \left \lfloor \frac{k}{r} \right \rfloor +2s+1 +\left \lceil  \frac{ (k-s) -\tau -\eta \cdot s }{\gamma}\right \rceil & \mbox{ if } r \nmid k 
\end{cases}$
 & $(r+1)\frac{k}{r} + 2\zeta, \mbox{ if } r|k $ \\
\hline
\end{tabular}}
\end{center}
\caption{Upper bounds for $\Bc(k,r,t)$ and $\Pc(k,r,t)$}
\label{table:results_upper}
\vspace{-.5cm}
\end{table*}

%
%
\section{Bounds on the dimension of a batch code}
\label{sec:ab}
Let $k_q^{\mbox{\footnotesize {\sf opt}}}(n,d)$ denote the largest possible dimension of a linear code of length $n$ and minimum distance $d$, for a given alphabet $\Qc$ of size $q$. More formally, by following on the notations in~\cite{CM}, denote:
$$k_q^{\mbox{\footnotesize {\sf opt}}}(n,d) \triangleq \max\frac{\log |\Cc|}{\log q} \; , $$ 
where the maximum is taken over all possible linear codes $\Cc$ of length $n$ with minimum distance $d$.

Let $\Ic \subseteq [n]$ be a set of coordinates. 
Define 
$$
H(\Ic)= \frac{\log|\{\bldx_{\Ic}: \bldx \in \Cc\}|}{\log q} \; .
$$ 
The following result appears as Lemma 2 in~\cite{CM}.
\begin{lemma}
Consider an $[n, k, d]$ code over $\Qc$ where there exists a set $\Ic \subseteq [n]$ such that $H(\Ic) \le m$. Then, 
there exists an $[n - |\Ic|, (k-m)^+, d]$ code over $\Qc$, where the $+$ symbol denotes the dimension is at least $k-m$.
\label{lemma:reduction-k}
\end{lemma}

Cadambe and Mazumdar show in~\cite{CM} that, for any $r$-locally recoverable $(n, k, d)$ code over the alphabet $\mathcal{Q}$, it holds: 
\begin{equation}
k \le \min_{t}\left[tr + k_q^{\mbox{\footnotesize {\sf opt}}}(n-t(r+1), d)\right] \; . 
\end{equation}

By using similar techniques, in the sequel we show a bound on $k$ 
for a linear $(n,k,r,t)$ batch code. We restrict our discussion to linear codes only. 
\begin{proposition}
\label{cadambe}
Let $\Cc$ be a linear $(n,k,r,t)$ batch code over an alphabet $\mathcal{Q}$ of size $q$ with minimum distance $d$, and $n - tr \ge d$. Then, 
\begin{equation}
 k \le tr - (t-1) + k_q^{\mbox{\footnotesize {\sf opt}}}(n-tr, d) \; . 
\label{eq:bound-k}
\end{equation}
\end{proposition}

\begin{proof}
Since $\Cc$ is an $(n,k,r,t)$ batch code, for a query $(x_{i},\ldots, x_{i})$ of size $t$, for some $i \in [k]$, there exist $t$ disjoint recovery sets $R_1,\ldots, R_t$ with $|R_j| \le r$ for all $j \in [t]$. 

Denote $\Ic \triangleq R_1 \cup R_2 \cup \cdots \cup R_t$. Clearly, $|\Ic| \le tr$. 
Take an arbitrary word $\bldy \in \Cc$. By Lemma~\ref{cor:sets}, there exist indices $a_2 \in R_2, a_3 \in R_3, \cdots, a_t \in R_t$ such that if we fix the values of all coordinates of $\bldy$ indexed by the
sets $R_1, R_2 \setminus \{a_2\}, R_3 \setminus \{a_3\}, \ldots ,R_{t} \setminus \{a_{t}\}$, then the values of the coordinates of $\bldy$ indexed by $S \triangleq \{a_2, a_3, \ldots, a_t\}$ are uniquely determined. 
It follows that there is one-to-one mapping between $\bldy_{\Ic \setminus S}$ and $\bldy_\Ic$.  
Therefore, $H(\Ic) \le tr - (t-1)$.

The main statement now follows by applying Lemma~\ref{lemma:reduction-k}. 
\end{proof}
We remark that the condition $n - tr \ge d$ in Proposition~\ref{cadambe} is necessary for existence of a code of length $n-tr$ and minimum distance $d$. However, it should be noted that any $(n,k,r,t)$ batch code is a $(n,k,r,\beta)$ batch code for $1 \le \beta \le t$. Thus, if $t$ is too large to satisfy this condition, we can take a smaller value $\beta$ instead. 
The following result follows. 
\begin{corollary}
\label{cadambe-2}
Let $\Cc$ be a linear $(n,k,r,t)$ batch code over an alphabet $\mathcal{Q}$ of size $q$ with minimum distance $d$. Then, 
\begin{equation}
 k \le \min_{1 \le \beta \le  t} 
\left\{ \beta r - (\beta-1) + k_q^{\mbox{\footnotesize {\sf opt}}}(n-\beta r, d) \right\} \; . 
\label{eq:bound-k-2}
\end{equation}
\end{corollary}
\subsection*{Comparison with the bounds in the literature}
\begin{example}
The asymptotic versions of the classical Singleton and the sphere-packing bounds for a code over alphabet $\Qc$ of size $q$ are
$R \le 1- \delta + o(1)$ and $R \le 1- h_q(\delta/2) +o(1)$, 
respectively, where $R= k/n,~\delta=d/n$, and $h_q(\cdot)$ denotes the $q$-ary entropy function~\cite[Chapter 4]{roth}.

The asymptotic versions of the bounds~(\ref{nonsyst}) and~(\ref{eq:bound-k-2}) can be rewritten as (when ignoring the $o(1)$ term, for any specific value of $\beta$):
\begin{equation}
\label{nonsystasym}
R \le 1-\delta -\frac{\beta -1}{n}\left(\left \lceil \frac{k}{\beta r -\beta+1}\right \rceil -1\right ) \; , 
\end{equation}
and
\begin{equation}
\label{singletonasym}
R \le \frac{\beta(r-1)}{n}+R_q^{\mbox{\footnotesize {\sf opt}}}(n-\beta r, d) \; , 
\end{equation}
respectively, where $R_q^{\mbox{\footnotesize {\sf opt}}}(n, d)$ denotes the maximum rate of any code of length $n$ and minimum distance $d$ over $\Qc$.  
	
For $n - d \gg tr \ge \beta r$, the bound (\ref{nonsystasym}) becomes  $R \le (1-\delta)\cdot\left( 1- \frac{\beta - 1}{\beta r}\right)$. For comparison, when we use the sphere-packing bound for $R_q^{\mbox{\footnotesize {\sf opt}}}(n, d)$, then (\ref{singletonasym}) can be rewritten as 
$$ R \le \frac{\beta(r-1)}{n}+1- h_q(\delta/2) \approx  1 - h_q(\delta/2) \; . $$

Therefore, for large blocklength, $n - d \gg tr$, and for a range of $\delta$ and $r$, the bound (\ref{eq:bound-k-2}) is tighter than (\ref{nonsyst}) (for any $\beta \ge 1$).
\end{example}

%
%




\section*{Acknowledgment}

The work of E. Thomas and V. Skachek is supported by the Estonian Research Council under the grants PUT405 and IUT2-1.
The authors wish to thank Mart Simisker  for valuable comments on the earliest version of this work.

\end{document}